\documentclass[12pt]{iopart}
\usepackage{iopams}
\usepackage{setstack}
\usepackage{graphicx}
\usepackage{amsthm}
\newtheorem{theorem}{Theorem}

\newtheorem{lemma}{Lemma}

\theoremstyle{remark}
\newtheorem*{remark}{Remark}

\begin{document}
\newcommand{\real}{\textrm{Re}\:}
\newcommand{\sto}{\stackrel{s}{\to}}
\newcommand{\supp}{\textrm{supp}\:}
\newcommand{\wto}{\stackrel{w}{\to}}
\newcommand{\ssto}{\stackrel{s}{\to}}
\newcounter{foo}
\providecommand{\norm}[1]{\lVert#1\rVert}
\providecommand{\abs}[1]{\lvert#1\rvert}

\title{Universal angular probability distribution of three particles near zero energy threshold}

\author{Dmitry K. Gridnev}

\address{FIAS, Ruth-Moufang-Stra{\ss}e 1, D--60438 Frankfurt am Main, Germany}
\ead{gridnev@fias.uni-frankfurt.de}
\begin{abstract}
We study bound states of a 3--particle system in $\mathbb{R}^3$ described by the Hamiltonian $H(\lambda_n) = H_0 + v_{12} + \lambda_n (v_{13} + v_{23})$, where
the particle pair $\{1,2\}$ has a zero energy resonance and no bound states, while other particle pairs have neither bound states nor zero energy resonances. 
It is assumed that for a converging sequence of coupling constants $\lambda_n \to \lambda_{cr}$ the Hamiltonian $H(\lambda_n)$ has a sequence of levels 
with negative energies $E_n$  and wave functions $\psi_n$, where the sequence $\psi_n$ totally spreads in the sense that 
$\lim_{n \to \infty}\int_{|\zeta| \leq R} |\psi_n (\zeta)|^2 d\zeta = 0$ for all $R>0$. We prove that for large $n$ the
angular probability distribution of three particles determined by $\psi_n$ 
approaches the universal analytical expression, which does not depend on pair--interactions. The result has applications in Efimov physics and in the physics of halo
nuclei.
\end{abstract}



\section{Introduction}\label{sec:1}

Consider the Hamiltonian of the 3--particle system in $\mathbb{R}^3$
\begin{equation}\label{hami}
 H(\lambda) = H_0 + v_{12} + \lambda (v_{13} + v_{23}) ,
\end{equation}
where $H_0$ is the kinetic energy operator with the center of mass removed,
$\lambda >0$ is the coupling constant and none of the particle
pairs has negative energy bound states. The detailed requirements on pair--potentials would be
listed in Sec.~\ref{sec:3}.
Suppose that for a converging sequence of coupling
constants $\lambda_n \to \lambda_{cr}$ there exists a sequence of bound states
$\psi_n \in D(H_0)$ such that $H(\lambda_n) \psi_n = E_n \psi_n$,
where $E_n < 0$, $\|\psi_n\| =1$ and $E_n \to 0$. The question, whether the
sequence $\psi_n$ totally spreads has been recently considered in \cite{1,2}.
In \cite{1} it was shown that $\psi_n$ does not spread if the 2--particle subsystems of $H(\lambda_n)$, $H(\lambda_{cr})$ have no zero energy resonances. 
The results of \cite{1} were generalized to many--particle systems \cite{2}, where,
in
particular, the restriction
on the sign of pair--potentials was removed. In \cite{1} under certain
conditions on pair--potentials it was proved that
 if the pair of particles $\{1,2\}$  has a zero energy resonance and $\psi_n$
for each $n$ is the ground state then the sequence $\psi_n$ totally spreads.

Here we focus again on the situation, where the pair of particles $\{1,2\}$  has
a zero energy resonance and the sequence $\psi_n (x,y)$
(\textit{not necessarily ground states!}) totally spreads. (For the definition
of Jacobi
coordinates $x,y \in \mathbb{R}^3$ see \cite{1} or Sec.~\ref{sec:3} of this paper). Recall that by definition in \cite{1} the total spreading means that 
\begin{equation}
 \lim_{n \to \infty} \int_{|x|^2 + |y|^2 \leq R}|\psi_n (x,y)|^2   \: d^3 x d^3 y\to 0 \quad (\textnormal{for all $R >0$}) . 
\end{equation}
Thereby, especially interesting is the angular probability distribution
of
three particles for large $n$, which we define below. 
Let us rewrite the wave function in the form $\psi_n (\rho , \theta ,  \hat x, \hat y) $, where the arguments are the so--called
hyperspherical coordinates \cite{babaev} $\rho := \sqrt{|x^2| + |y|^2}$,
$\theta :=
\arctan (|y|/|x|) $,
$\theta \in [0 , \pi/2]$ and $\hat x , \hat y$ are unit vectors in the directions
of $x,y$ respectively. Then by definition the angular probability distribution is
\begin{equation}\label{xw2}
 \mathcal{D}_n (\theta , \hat x , \hat y) := \cos^2 \theta \sin^2 \theta\int
\rho^5 \;  \bigl|\psi_n (\rho , \theta , \hat x, \hat y) \bigr|^2 d \rho .
\end{equation}
The normalization  $\|\psi_n \| = 1$ implies that
\begin{equation}
 \int_0^{\pi/2} d\theta \int d \Omega_x  \int d\Omega_y \; \mathcal{D}_n (\theta
, \hat x , \hat y) = 1 , 
\end{equation}
where $\Omega_{x , y}$ are the body angles associated with the unit vectors
$\hat x , \hat y$.
The main result of the present paper (proved in Theorem~\ref{th:3}) states that
\begin{equation}\label{sony}
 \mathcal{D}_\infty  (\theta , \hat x , \hat
y) := \lim_{n \to \infty}\mathcal{D}_n (\theta , \hat x , \hat
y) = \frac 1{4\pi^3} \sin^2 \theta , 
\end{equation}
where the convergence is in measure. Equation (\ref{sony}) means that all
acceptable pair--potentials produce the same limiting angular probability distribution, which 
depends solely on $\theta$. This is another example of the so--called universality in three--particle systems, which is, in particular, manifested 
in the universal asymptotic form of the infinite discrete spectrum appearing the Efimov effect \cite{yafaev,sobol}. 
Apart from the results in \cite{1,2} the proof resides on the ideas expressed in \cite{yafaev,sobol,klaus1,klaus2}. 
In the next section we shall discuss the two--particle case, this material would 
also be needed in the analysis of the three--particle case in Sec.~\ref{sec:3}. 
At the end of Sec.~\ref{sec:3} we show how the distribution in (\ref{sony}) can be derived fairly easy on a physical level of rigor (this derivation was proposed by one of the referees). 
In Sec.~\ref{sec:phys} we discuss physical applications. 

\section{The Two-Particle Case Revisited}\label{sec:2}

 Let us consider the two--particle Hamiltonian in $L^2(\mathbb{R}^3)$
\begin{equation}\label{2h}
 h(\lambda) = - \Delta_x + \lambda v(x),
\end{equation}
where $\lambda >0 $ is a coupling constant.
For the pair potential we assume that 
\begin{equation}\label{con'}
\gamma := \max\left[ \int d^3x \;  |x|^2 \bigl( 1+ |x|^\delta \bigr) |v(x)|^2 , \int d^3x \;
 \bigl(1+ |x|^\delta \bigr) |v(x)|^2 \right] < \infty, 
\end{equation}
where $0 < \delta < 1$ is some constant.

The next theorem (which must be known in some form) states that a totally
spreading sequence of bound state wave functions approaches the expression, which is independent of
the details
of the pair--interaction.
\begin{theorem}\label{th:1}
 Suppose there is a sequence of coupling constants $\lambda_n >0$
such that  $\lim_{n\to \infty} \lambda_n =
\lambda_{cr} >0 $, and
$h(\lambda_n) \psi_n = E_n \psi_n$, where $\psi_n \in D(H_0)$, $\|\psi_n\| =1$,
$E_n < 0$, $\lim_{n \to \infty} E_n = 0$. If $\psi_n$  totally spreads then
\begin{equation}\label{tbh}
 \left\| \psi_n -
e^{i\varphi_n}\frac{\sqrt{k_n}e^{-k_n|x|}}{\sqrt{2\pi}|x|}\right\| \to 0 ,
\end{equation}
where $\varphi_n \in [0, 2\pi)$ are phases and $k_n := \sqrt{|E_n|}$.
\end{theorem}
A few remarks are in order. If one takes for $\psi_n$ the ground states then the
sequence $\psi_n$ always totally spreads, see the discussion in \cite{klaus1,garcia}. In the spherically
symmetric potential
$s$--states always spread, and states, which have a non-zero angular momentum, do not spread \cite{klaus1}. (This can also be
seen from (\ref{tbh}), which tells that the wave function must approach the spherically symmetric form). Let us also note that $\psi_n$ does not spread
if $v(x) \geq |x|^{-2 + \epsilon}$ for $|x| \geq R_0$ and $\epsilon\in (0,1)$,
see \cite{garcia,bolle,bollenew}.

\begin{proof}[Proof of Theorem~\ref{th:1}]
Obviously, $R_n := (\psi_n , (1 + |x|^\delta)^{-1} \psi_n) \to 0 $ because $\psi_n$ totally
spreads.
The Schr\"odinger
equation in the integral form reads
\begin{equation}
 \tilde \psi_n = \frac {\lambda_n}{4\pi} \int d^3x'\;  \frac{e^{-k_n
|x-x'|}}{|x-x'|} v (x')\tilde \psi_n (x') ,
\end{equation}
where $\tilde \psi_n := \psi_n /R^{1/2}_n$ is the renormalized wave function . 
Let us set
\begin{equation}\label{f(n)}
f_n := \frac {\lambda_n}{4\pi} \frac{e^{-k_n |x|}}{|x|} \int d^3x'\;   v
(x')\tilde \psi_n (x') . 
\end{equation}
Our aim is to prove that $\| \tilde \psi_n - f_n \| = \mathcal{O}(1)$. The
direct calculation gives
\begin{eqnarray}
\| \tilde \psi_n - f_n \|^2 = \frac {\lambda^2_n}{16\pi^2} \int d^3 x d^3 x' d^3
x'' \left[\frac{e^{-k_n |x-x'|}}{|x-x'|} - \frac{e^{-k_n |x|}}{|x|} \right]
\nonumber\\
\times \left[\frac{e^{-k_n |x-x''|}}{|x-x''|} - \frac{e^{-k_n |x|}}{|x|} \right] v(x') v (x'') \tilde \psi^*_n (x') \tilde \psi_n (x'') . 
\end{eqnarray}
This can be transformed into
\begin{eqnarray}
\| \tilde \psi_n - f_n \|^2 = \frac {\lambda^2_n}{16\pi^2} \int d^3 x' d^3 x''
\frac 1k_n \Bigl\{ W(k_n (x'' -x')) + W(0)
\nonumber \\
- W(k_n x') - W(k_n x'')\Bigr\} v(x') v (x'') \tilde \psi^*_n (x') \tilde \psi_n (x'') \label{77},
\end{eqnarray}
where we defined
\begin{equation}\label{W(y)}
 W (y) := \int d^3z \; \frac{e^{-|z|} e^{-|z-y|}}{|z|\; |z-y|}  = 2\pi e^{-|y|} . 
\end{equation}
The integral in (\ref{W(y)}) can  be evaluated using the confocal elliptical
coordinates, see f. e. Appendix~9 in \cite{bransden}. Next, by the obvious
inequality $|W(y) -
W(0)| \leq 2\pi |y|$
\begin{eqnarray}
\| \tilde \psi_n - f_n \|^2 \leq \frac {\lambda^2_n }{8\pi}  \int d^3 x' d^3 x''
\bigl\{ |x'' -x'| + | x'| + | x''|\bigr\}  |v(x')| |v (x'')|  \nonumber \\
\times |\tilde \psi_n (x')|
|\tilde \psi_n (x'')|\leq \frac {\lambda^2_n }{2\pi}  \int d^3 x' d^3 x'' | x'|  |v(x')| |v (x'')|
|\tilde \psi_n (x')| |\tilde \psi_n (x'')| \label{xw1}. 
\end{eqnarray}
Inserting into the rhs of (\ref{xw1}) the identities $1 = (1+ |x'|^\delta)^{1/2} (1+
|x'|^\delta)^{-1/2} $ and  the same for $x''$ and applying the Cauchy--Schwarz
inequality gives
\begin{equation}
\| \tilde \psi_n - f_n \|^2 \leq  \frac {\lambda^2_n \gamma}{2\pi}  ,
\end{equation}
where $\gamma$ is defined in (\ref{con'}).
Thus $\| \tilde \psi_n - f_n \| = \mathcal{O}(1)$ and by (\ref{f(n)}) we have
\begin{equation}
 \psi_n = \frac{\lambda_n}{4\pi} R^{1/2}_n d_n \frac{e^{-k_n |x|}}{|x|} +
\hbox{o}(1),
\end{equation}
where $d_n := \int d^3x'\;   v (x')\tilde \psi_n (x')$ and $\hbox{o}(1)$ denotes
the terms that go to zero in norm. Using that $\|\psi_n\| = 1$
we recover the statement of the theorem. 
 \end{proof}

\section{The Three--Particle Case}\label{sec:3}
We shall consider the Hamiltonian (\ref{hami}). Let $m_i$ and
$r_i \in \mathbb{R}^3$ denote particles masses and position vectors. The reduced
masses we shall denote as $\mu_{ik} := m_i m_k / (m_i + m_k)$.
The pair--interactions $v_{ik}$ are operators of multiplication
by real $V_{ik} (r_i - r_k)$. We shall make the following assumptions

\begin{list}{R\arabic{foo}}
{\usecounter{foo}
    \setlength{\rightmargin}{\leftmargin}}
\item
The pair potentials satisfy the following requirement
\begin{equation}
\fl \gamma_0 := \max_{i=1,2}\max\left[ \int d^3 r \bigl| V_{i3} (r)\bigr|^2 , \int
d^3 r \bigl| V_{i3} (r)\bigr| (1+|r|)^{2\delta}\right] < \infty , \label{restr}
\end{equation}
where $0 < \delta < 1/8$ is a fixed constant. And
\begin{equation}\label{restr3}
 -b_1 e^{-b_2 |r|} \leq V_{12} (r) \leq 0 ,
\end{equation}
where $b_{1,2} >0$ are some constants.
\item
There is a converging sequence of coupling constants $\lambda_n >0$, 
 $\lim_{n\to \infty} \lambda_n =
\lambda_{cr} >0$  such that
$H(\lambda_n) \psi_n = E_n \psi_n$, where $\psi_n \in D(H_0)$, $\|\psi_n\| =1$,
$E_n < 0$, $\lim_{n \to \infty} E_n = 0$.
\item
The Hamiltonian $H_0 + v_{12}$ is at critical coupling (For the definition of critical coupling see \cite{2}).
The Hamiltonians $H_0 + \lambda v_{13}$ and $H_0 + \lambda v_{23}$ are positive
and are not at critical coupling for $\lambda = \lambda_n , \lambda_{cr}$.
\end{list}

Again, let us stress that given that R1 is satisfied one can always tune the coupling constants so that R2, R3 would be satisfied with $\psi_n$ being ground states. 
Besides, the sequence $\psi_n$ in this case would totally spread, this is discussed in detail in Sec.~6 in \cite{1}. 

In the Jacobi coordinates $x := [\sqrt{2 \mu_{12}}/\hbar](r_2 - r_1)$ and $y :=
[\sqrt{2 M_{12}}/\hbar](r_3
- m_1/(m_1+m_2) r_1 - m_2/(m_1+m_2) r_2)$, where
 $M_{ij} = (m_i + m_j)m_k / (m_1 + m_2 + m_3)$ ($\{i,j,k\}$ is a permutation of $\{1,2,3\}$) the kinetic
energy operator takes the form \cite{1,2}
\begin{equation}\label{ay4}
    H_0 = - \Delta_x- \Delta_y .
\end{equation}
In the following $\chi_\Omega : \mathbb{R} \rightarrow \mathbb{R} $
denotes
the characteristic function of the interval $\Omega \subset \mathbb{R} $ (for instance, $\chi_{[1, \infty)} (x)$ is equal to one if $x \in [1, \infty)$ and is zero otherwise). The
next theorem is the analog of Theorem~\ref{th:1} for the three--particle case.
\begin{theorem}\label{th:2}
Suppose $H(\lambda)$ defined in (\ref{hami}) satisfies R1--3. If $\psi_n$
totally spreads then
\begin{equation}\label{3bh3}
 \left\| \psi_n -  \frac{e^{i\varphi_n}\chi_{[1,\infty)}(\rho)}{2\pi^{3/2} |\ln
k_n|^{1/2}} \frac{ \bigl\{ |x|\sin(k_n |y|) + |y|\cos(k_n |y|)\bigr\}e^{-k_n|x|}
}{|x|^3|y|+|y|^3 |x|} \right\| \to 0 , 
\end{equation}
where $\varphi_n \in [0, 2\pi)$ are phases, $\rho := \sqrt{|x|^2 + |y|^2}$ and
$k_n
:= \sqrt{|E_n|}$.
\end{theorem}
\begin{remark}
Theorem~\ref{th:2} shows that similar to the two--particle case total spreading in the considered three--body case is possible only for states with zero angular momentum 
(irrespectively of the values of particles' masses). This can be seen 
from (\ref{3bh3}), where in the limit the wave function depends only on $|x|,|y| $ and is thus invariant under rotations of $x,y$. This fact is not that trivial as it may seem. 
In Theorem~\ref{th:2} we consider only the situation when a single particle pair has a zero energy resonance. If two particle pairs have such resonances, one arrives at the Efimov effect, see 
\cite{8.efimov,yafaev,sobol}, where there exists an infinite sequence of energy levels $E_n \to 0$ with orthonormal wave functions $\phi_n$. This sequence of wave functions also totally spreads, 
see \cite{2}. However, in the case of Efimov effect it is possible to choose mass ratios in the system in such a way that the sequence 
$\phi_n$ would have a non--zero angular momentum, see \cite{hwhammer,efimnucl,sobol}.  
\end{remark}

Theorem~\ref{th:2} has a useful practical corollary.
\begin{theorem}\label{th:3}
 Suppose $H(\lambda)$ satisfies R1--3. If $\psi_n$
totally spreads then the angular probability distribution $\mathcal{D}_n (\theta , \hat x , \hat y)$
defined in (\ref{xw2}) converges in measure to
$\mathcal{D}_\infty (\theta , \hat x , \hat y) = (4\pi^3)^{-1} \sin^2 \theta$.
\end{theorem}
\begin{proof}
 Let us rewrite (\ref{3bh3}) in hyperspherical coordinates
\begin{equation}\label{3bh4}
 \left\| \psi_n - \Theta_n \right\| \to 0 ,
\end{equation}
where
\begin{equation}\label{17:23:1}
\Theta_n:= \frac{e^{i\varphi_n} \chi_{[1,\infty)}(\rho)}{2\pi^{3/2}
|\ln k_n|^{1/2}} \frac{e^{-k_n\rho\cos\theta} \sin(\theta + k_n \rho
\sin\theta)}{\rho^3 \cos\theta \sin\theta} . 
\end{equation}
If we denote by $\mathcal{D}^{\Theta}_n (\theta, \hat x , \hat y)$ the angular
probability distribution given by $\Theta_n$ then the limiting angular probability distribution is
\begin{eqnarray}
\fl \mathcal{D}_\infty (\theta, \hat x , \hat y)
= \lim_{n\to \infty} \mathcal{D}^{\Theta}_n = \frac{1}{4\pi^3}\lim_{n \to \infty} \frac{1}{|\ln k_n|}\int_1^\infty
\frac{e^{-2k_n\rho\cos\theta}}{\rho} \sin^2 (\theta + k_n \rho \sin \theta) \, d\rho , 
\end{eqnarray}
where the limit is pointwise. 
Changing the integration variable in the last integral for $t = k_n \rho \sin
\theta$ and expanding around $t=0$ we obtain
\begin{eqnarray}
\fl \mathcal{D}_\infty (\theta, \hat x , \hat y)  = \frac{1}{4\pi^3} \lim_{n \to
\infty}\frac{1}{|\ln k_n|}\int_{k_n \sin
\theta}^\infty \frac{e^{-2t\cot\theta}}{t} \sin^2 (\theta + t) \, dt  =
\frac 1{4\pi^3}\sin^2 \theta . 
\end{eqnarray}
Note that $\mathcal{D}^\Theta_n \to \mathcal{D}_\infty$ pointwise and uniformly. Now we show that $\| \mathcal{D}_n - \mathcal{D}^{\Theta}_n \|_1  \to 0$. To
make the notation shorter we set
$d \tilde \Omega := \cos^2\theta \sin^2 \theta d\theta d\Omega_x d \Omega_y$.
\begin{eqnarray}
 \| \mathcal{D}_n - \mathcal{D}^{\Theta}_n \|_1 \equiv \int_0^{\pi/2} d\theta
\int d\Omega_x d\Omega_y \bigl| \mathcal{D}_n - \mathcal{D}^{\Theta}_n \bigr| \nonumber\\
\fl  =
\int d\tilde \Omega \Bigl| \int |\psi_n|^2 \rho^5 d\rho - \int |\Theta_n |^2
\rho^5 d\rho\Bigr| 
\leq \int d\tilde \Omega \int \rho^5 \Bigl| |\psi_n| - |\Theta_n| \Bigr|
\Bigl(  |\psi_n |   + |\Theta_n |  \Bigr) \,  d\rho  \nonumber \\
\leq \bigl\|\psi_n - \Theta_n \bigr\| \left( \int d\tilde \Omega \int \rho^5
  \bigl( |\psi_n |   + |\Theta_n | \bigr)^2 \,  d\rho \right)^{1/2} \leq 2
\bigl\|\psi_n - \Theta_n \bigr \| , \label{14:12:1}
\end{eqnarray}
where we applied twice the Cauchy--Schwarz inequality and $\bigl| |a| - |b|
\bigr| \leq |a-b|$ for any $a,b \in \mathbb{C}$. Therefore,
$\| \mathcal{D}_n - \mathcal{D}_\infty \|_1 \to 0$. By the Vitali convergence
theorem \cite{bartle} this is equivalent to the statement of Theorem~\ref{th:3}.
 \end{proof}
Here we would like to make the following two remarks.  
\begin{remark}
If instead of Jacobi coordinates one would express the limiting angular
probability distribution in $r_{13} :=r_3-r_1$ and
$r_{23} := r_3 - r_2$, which are also ``natural'' coordinates for the considered
problem, then
it would depend not only on the ratio $|r_{13}|/|r_{23}|$ but also on the angle
between these vectors. Let us also note that if the pair of particles $\{1,2\}$
would
be marginally bound with the energy $E_{12}$ and the sequence of ground states
$\psi_n$ would be such that $E_n < E_{12}$, $E_n \to E_{12}$ then $\psi_n $
totally spreads,
see \cite{klaus2}. However, in this case it is easy to show that the angular
probability distribution approaches the delta--distribution.
\end{remark}
\begin{remark}
Theorem~\ref{th:3} states that the angular probability distribution converges in measure, which is equivalent to $\|\mathcal{D}_n - \mathcal{D}_\infty\|_1 \to 0$, whereby 
$\mathcal{D}_n$ corresponds to $\psi_n$ in R2 and the meaning of $\|\cdot\|_1$ is explained in (\ref{14:12:1}). It should be stressed here that 
$\mathcal{D}_n$ does not converge to $\mathcal{D}_\infty$ pointwise everywhere. Indeed, for smooth interactions one expects that $\psi_n$ at $|x| = 0$ should be finite; by definition 
(\ref{xw2}) this immediately implies that $\mathcal{D}_n = 0$ if $\theta = \pi/2$. Nevertheless, the limiting angular probability distribution $\mathcal{D}_\infty$ 
is not zero at $\theta = \pi/2$, on the contrary, it has its maximum at this point! At $\theta \neq \pi/2$ one should expect a pointwise convergence. 
\end{remark}

\begin{lemma}\label{lem:1}
 Suppose $H(\lambda)$ defined in (\ref{hami}) satisfies R1--3. If $\psi_n$
totally spreads then
\begin{equation}
 \psi_n = \bigl[H_0 + k_n^2\bigr]^{-1} | v_{12} | \psi_n  + \hbox{o}(1),
\end{equation}
where $\hbox{o}(1)$ denotes the terms that go to zero
in norm.
\end{lemma}
\begin{proof}
Rearranging in different ways the terms in the Schr\"odinger equation for $\psi_n$ we derive three
 integral equations, see \cite{2}
\begin{eqnarray}
 \fl \psi_n = \bigl[H_0 + k_n^2\bigr]^{-1} \Bigl( |v_{12}| - \lambda_n v_{13} -
\lambda_n v_{23} \Bigr)\psi_n \label{ihami} , \\
 \fl \psi_n = \bigl[ H_0 + \lambda_n (v_{13})_+ + \lambda_n (v_{23})_+  +
k_n^2\bigr]^{-1} \Bigl( |v_{12}| + \lambda_n (v_{13})_- 
+ \lambda_n (v_{23})_-\Bigr)\psi_n  \label{ihami2} ,  \\
\fl \psi_n = \bigl[ H_0 + \lambda_n (v_{13})_+ + k_n^2\bigr]^{-1} \Bigl( |v_{12}| +
\lambda_n (v_{13})_- - \lambda_n v_{23}  \Bigr)\psi_n  , \label{ihami3}
\end{eqnarray}
where $(v_{ik})_\pm = \max [0, \pm v_{ik}]$.
By (\ref{ihami}) the Lemma would be proved if we can show that
\begin{eqnarray}
 F_n :=  \lambda_n \bigl[H_0 + k_n^2\bigr]^{-1} v_{13} \psi_n =  \hbox{o}(1) , 
\label{kost1}\\
\lambda_n \bigl[H_0 + k_n^2\bigr]^{-1} v_{23} \psi_n =  \hbox{o}(1) . 
\label{kost2}
\end{eqnarray}
Below we prove (\ref{kost1}), eq.~(\ref{kost2}) is proved analogously.
Substituting (\ref{ihami2}) into (\ref{kost1}) we split $F_n$ in three parts
\begin{equation}\label{xwsplit}
 F_n = \sum_{i=1}^3 F^{(i)}_n , 
\end{equation}
where
\begin{eqnarray}
 \fl F^{(1)}_n = \bigl[ H_0 + k_n^2\bigr]^{-1} v_{13} \bigl[ H_0 + \lambda_n
(v_{13})_+ + \lambda_n (v_{23})_+  + k_n^2\bigr]^{-1} |v_{12}|\psi_n  , \label{xwF1n}\\
\fl F^{(2)}_n = \lambda_n \bigl[ H_0 + k_n^2\bigr]^{-1} v_{13} \bigl[ H_0 + \lambda_n
(v_{13})_+ + \lambda_n (v_{23})_+  + k_n^2\bigr]^{-1} (v_{23})_- \psi_n  ,  \label{xwF2n}\\
\fl F^{(3)}_n = \lambda_n  \bigl[ H_0 + k_n^2\bigr]^{-1} v_{13} \bigl[ H_0 + \lambda_n
(v_{13})_+ + \lambda_n (v_{23})_+  + k_n^2\bigr]^{-1} (v_{13})_- \psi_n  .
\label{F3}
\end{eqnarray}

We introduce another pair of Jacobi coordinates
$\eta = [\sqrt{2 \mu_{13}}/\hbar](r_3 - r_1)$ and
$\zeta = [\sqrt{2 M_{13}}/\hbar](r_2 - m_1/(m_1+m_3) r_1 - m_3/(m_1+m_3) r_3)$.
The coordinates $(\eta , \zeta)$ and $(x,y)$ are connected through the orthogonal linear
transformation
\begin{eqnarray}
 x = m_{x\eta} \eta +  m_{x\zeta} \zeta , \\
 y = m_{y\eta} \eta +  m_{y\zeta} \zeta,
\end{eqnarray}
where $m_{x\eta}, m_{x\zeta} \neq 0, m_{y\eta}, m_{y\zeta}$
are real and can be expressed through mass ratios in the
system.
$\mathcal{F}_{13}$ denotes the partial Fourier transform, which acts on $f(\eta,
\zeta) $ as
\begin{equation}\label{xw20}
 \mathcal{F}_{13} f := \hat f(\eta, p_\zeta) = \frac 1{(2 \pi )^{3/2}} \int d^3
\zeta \; e^{-ip_\zeta \cdot \; \zeta} f(\eta, \zeta) .
\end{equation}
Let us introduce the operator function
\begin{equation}\label{b13}
 \tilde B_{13} (k_n) := \mathcal{F}^{-1}_{13} \tilde t_n (p_\zeta)
\mathcal{F}_{13} ,
\end{equation}
where
\begin{equation}\label{tail}
    \tilde t_n (p_\zeta) = \left\{ \begin{array}{ll}
    |p_\zeta|^{1-\delta} + (k_n)^{1-\delta}  & \quad \mathrm{if} \;\; |p_\zeta|
\leq 1  \\
    1 + (k_n)^{1-\delta} & \quad \mathrm{if} \;\; |p_\zeta| \geq 1 . \\
    \end{array}
    \right.
\end{equation}

We set tilde over the operator in order to distinguish it from the one defined
in Eq.~(18) in \cite{1}. Note that  $\tilde B_{13} (k_n)$ and $\tilde
B^{-1}_{13} (k_n)$ for each $n$ are bounded
operators.

Using the inequalities from \cite{2} (see Eqs.~(17)--(24) in \cite{2}) we obtain
\begin{eqnarray}
 \fl  |F^{(1)}_n| \leq \bigl[ H_0 + k_n^2\bigr]^{-1} |v_{13}| \bigl[ H_0 +
k_n^2\bigr]^{-1} |v_{12}| |\psi_n| 
 = \bigl[ H_0 + k_n^2\bigr]^{-1}
|v_{13}|^{1/2}
\tilde B_{13}(k_n) \Psi^{(1)}_n  , \label{xw3}\\
 \fl  |F^{(2)}_n| \leq \lambda_n \bigl[ H_0 + k_n^2\bigr]^{-1} |v_{13}| \bigl[ H_0 +
k_n^2\bigr]^{-1} |v_{23}| |\psi_n| = 
\bigl[ H_0 + k_n^2\bigr]^{-1}
|v_{13}|^{1/2}
\tilde B_{13}(k_n) \Psi^{(2)}_n \label{xw4},
\end{eqnarray}
where
\begin{eqnarray}
 \Psi^{(1)}_n  := |v_{13}|^{1/2} \tilde B_{13}^{-1} (k_n) \bigl[ H_0 +
k_n^2\bigr]^{-1} |v_{12}| |\psi_n| , \label{psin1}\\
 \Psi^{(2)}_n  := \lambda_n |v_{13}|^{1/2} \tilde B_{13}^{-1} (k_n) \bigl[ H_0 +
k_n^2\bigr]^{-1} |v_{23}| |\psi_n|  \label{psin2}.
\end{eqnarray}
To write the upper bound on $ |F^{(3)}_n|$ we use the following expression, which
follows from (\ref{ihami3}), c. f. Eq.~(15) in \cite{2}
\begin{equation}\label{gemor}
 (v_{13})^{1/2}_- \psi_n = Q_n (v_{13})^{1/2}_- \bigl[ H_0 + \lambda_n
(v_{13})_+ +k_n^2\bigr]^{-1} \bigl(|v_{12}| - \lambda_n v_{23}\bigr) \psi_n ,
\end{equation}
where we defined
\begin{equation}
 Q_n := \Bigl\{1 - \lambda_n (v_{13})^{1/2}_- \bigl[H_0 + \lambda_n (v_{13})_+
+ k_n^2\bigr]^{-1} (v_{13})^{1/2}_- \Bigr\}^{-1} .
\end{equation}
$Q_n $ is a positivity preserving operator and $\sup_n \| Q_n \| < \infty$, see
Lemma~1 in \cite{2} and Lemma~12 in \cite{1}. 
Substituting (\ref{gemor}) into (\ref{F3}) and using the positivity preserving
property of the operators (see the discussion after Eq.~(16) in \cite{2}) we get
\begin{eqnarray}
 |F^{(3)}_n| \leq \lambda_n \bigl[ H_0 + k_n^2\bigr]^{-1} |v_{13}| \bigl[H_0 +
k_n^2\bigr]^{-1} (v_{13})^{1/2}_- Q_n
(v_{13})^{1/2}_- \nonumber \\
\times \bigl[ H_0 +k_n^2\bigr]^{-1} \bigl(|v_{12}| +  \lambda_n
|v_{23}|\bigr)|\psi_n| 
= \bigl[ H_0 + k_n^2\bigr]^{-1} |v_{13}|^{1/2}
\tilde B_{13}(k_n) \Psi^{(3)}_n \label{xw5},
\end{eqnarray}
where
\begin{eqnarray}
 \Psi^{(3)}_n  := \lambda_n  |v_{13}|^{1/2} \bigl[ H_0 + k_n^2\bigr]^{-1}  (v_{13})^{1/2}_-
  Q_n  \tilde B_{13}^{-1} (k_n)   (v_{13})^{1/2}_-  \nonumber \\
\times \bigl[ H_0 + k_n^2\bigr]^{-1}  \bigl(|v_{12}| +  \lambda_n
|v_{23}|\bigr)|\psi_n| . \label{psin3}
\end{eqnarray}
Summarizing, (\ref{xw3}), (\ref{xw4}) and (\ref{xw5}) can be expressed through the inequality
\begin{equation}\label{18.04/1}
  |F^{(i)}_n| \leq \mathcal{L}_n  \Psi^{(i)}_n \quad \quad (i=1,2,3),
\end{equation}
where
\begin{equation}\label{18.04/2}
\mathcal{L}_n  :=  \bigl[ H_0 + k_n^2\bigr]^{-1} |v_{13}|^{1/2}
\tilde B_{13}(k_n) . 
\end{equation}
From Lemma~\ref{lem:2} it follows that $\| F^{(i)}_n \| \to 0$.
 \end{proof}
\begin{lemma}\label{lem:2}
 The operators $\mathcal{L}_n $ are uniformly norm--bounded and
$\|\Psi^{(i)}_n\| \to 0$ for $i=1,2,3$.
\end{lemma}
\begin{proof}
 The proof  that  $\mathcal{L}_n $ are
uniformly norm--bounded is similar to Lemma~6 in \cite{1}. Indeed, 
$K_n := \mathcal{F}_{13}\mathcal{L}_n \mathcal{F}_{13}^{-1}$ is an integral operator with the kernel
\begin{equation}
 k_n (\eta , \eta';p_\zeta) = \frac{e^{-\sqrt{p_\zeta^2 +
k_n^2}|\eta-\eta'|}}{4\pi|\eta - \eta'|} \bigl| V_{13} (\alpha'
\eta')\bigr|^{1/2} \tilde t_n (p_\zeta) , 
\end{equation}
where $\alpha' := \hbar /\sqrt{2\mu_{13}} $, which acts on $f(\eta, p_\zeta) \in L^2 (\mathbb{R}^6)$ as follows 
\begin{equation}
K_n f = \int d^3 \eta' k_n (\eta , \eta';p_\zeta) f(\eta' , p_\zeta).
\end{equation}
Therefore, we can estimate the norm as 
\begin{equation}\label{xw7}
\|\mathcal{L}_n \|^2 = \| K_n \|^2 \leq \sup_{p_\zeta } \int \bigl| k_n (\eta , \eta';p_\zeta) \bigr|^2 \; d^3 \eta' d^3 \eta = C_0 \sup_{p_\zeta } \frac{|\tilde t_n (p_\zeta)|^2}{\sqrt{p_\zeta^2 +
k_n^2}},
\end{equation}
where
\begin{equation}
 C_0 := \frac 1{16\pi^2} \left( \int \frac{e^{-2|s|}}{|s|^2} d^3 s\right)   \left( \int \bigl| V_{13} (\alpha'
\eta)\bigr| d^3 \eta \right) \leq \frac {\gamma_0}{8\pi} , 
\end{equation}
and $\gamma_0$ was defined in (\ref{restr}). Substituting (\ref{tail}) into (\ref{xw7}) it is easy to see that $\|\mathcal{L}_n \|$ is uniformly bounded. 
Let us rewrite (\ref{psin1}) as
\begin{equation}
  \Psi^{(1)}_n  := \bigl[ \mathcal{M}^{(1)}_n + \mathcal{M}^{(2)}_n \bigr] |v_{12}|^{1/2}  |\psi_n| ,
\end{equation}
where
\begin{eqnarray}
\mathcal{M}^{(1)}_n := |v_{13}|^{1/2} \Bigl\{ \tilde B_{13}^{-1} (k_n) - \bigl(1+ (k_n)^{1-\delta}\bigr)^{-1} \Bigr\} \bigl[ H_0 + k_n^2\bigr]^{-1} |v_{12}|^{1/2}  , \\
\mathcal{M}^{(2)}_n := \bigl( 1+ (k_n)^{1-\delta} \bigr)^{-1}  |v_{13}|^{1/2} \bigl[ H_0 + k_n^2\bigr]^{-1} |v_{12}|^{1/2} . 
\end{eqnarray}
 By the no--clustering theorem $\bigl\||v_{12}|^{1/2}  |\psi_n|\bigr\| \to 0$, see Appendix in \cite{2}. Thus to prove that $\bigl\| \Psi^{(1)}_n \bigr\| \to 0$
it is enough to show that $\sup_n \bigl\|\mathcal{M}^{(1,2)}_n \bigr\| < \infty$. It is easy to see that $\| \mathcal{M}^{(2)}_n \|$ is uniformly norm--bounded, see f.~e. the proof of 
Lemma~7 in \cite{1}. Next, $\| \mathcal{M}^{(1)}_n \| = \| K'_n \| $, where
$K'_n := \mathcal{F}_{13}\mathcal{M}_n \mathcal{F}_{13}^{-1}$ is the integral operator with the kernel
\begin{eqnarray}
k'_n (\eta, p_\zeta, \eta', p'_\zeta ) = \frac 1{2^{7/2}\pi^{5/2} \omega^3} \Bigl[ \tilde t^{-1}_n (p_\zeta) - \bigl(1+ (k_n)^{1-\delta}\bigr)^{-1} \Bigr] \bigl| V_{13}  (\alpha' \eta)\bigr|^{1/2} \nonumber \\
\times \frac{e^{-\sqrt{p_\zeta^2 + k_n^2} |\eta-\eta'|}}{|\eta-\eta'|}
\exp{\left\{i\frac{\beta}{\omega} \eta' \cdot (p_\zeta - p'_\zeta)\right\}}\:
\widehat{\bigl| V_{12}\bigr|^{1/2}} \bigl( (p_\zeta - p'_\zeta)/\omega \bigr) ,  \label{four27}
\end{eqnarray}
$\beta := -m_3 \hbar / ((m_1 + m_3)\sqrt{2\mu_{13}})$
and $\omega := \hbar/\sqrt{2M_{13}}$
(see the proof of Lemma~9 in \cite{1}). In (\ref{four27}) $\widehat{\bigl| V_{12}\bigr|^{1/2}}$ denotes merely
the Fourier transform of $\bigl| V_{12}\bigr|^{1/2} \in L^2 (\mathbb{R}^3)$. Calculation of the Hilbert--Schmidt norm gives
\begin{eqnarray}
 \fl  \| \mathcal{M}^{(1)}_n \|^2 \leq \frac{C'_0}{8 \omega^3 \pi^3} \int_{|p_\zeta| \leq 1}  \frac{\bigl[|p_\zeta|^{1-\delta} + (k_n)^{1-\delta}\bigr]^{-2} }{\sqrt{p_\zeta^2 + k_n^2}} d^3 p_\zeta 
\leq \frac{C'_0}{8 \omega^3 \pi^3} \int_{|p_\zeta| \leq 1}  \frac {d^3 p_\zeta}{|p_\zeta|^{3-2\delta}} \label{xw9},
\end{eqnarray}
where
\begin{equation}
 C'_0 := C_0 \int d^3 s \Bigl| \widehat{\bigl| V_{12}\bigr|^{1/2}} (s)\Bigr|^2 . 
\end{equation}
From (\ref{xw9}) it follows that $\| \mathcal{M}^{(1)}_n \|$ is uniformly bounded and, hence, $\| \Psi^{(1)}_n \| \to 0$. The fact that $\| \Psi^{(2)}_n \| \to 0$
is proved analogously. To prove that $\| \Psi^{(3)}_n \| \to 0$ let us look at (\ref{psin3}). We can write
\begin{equation}
 \Psi^{(3)}_n = \lambda_n \mathcal{T}^{(1)}_n Q_n \bigl(\mathcal{T}^{(2)}_n |v_{12}|^{1/2} |\psi_n| +  \mathcal{T}^{(3)}_n |v_{23}|^{1/2} |\psi_n|\bigr), \label{psin32}
\end{equation}
where we defined the operators
\begin{eqnarray}
 \mathcal{T}^{(1)}_n := |v_{13}|^{1/2} \bigl[ H_0 + k_n^2\bigr]^{-1}  (v_{13})^{1/2}_-  , \\
 \mathcal{T}^{(2)}_n := \tilde B_{13}^{-1} (k_n) (v_{13})^{1/2}_- \bigl[ H_0 + k_n^2\bigr]^{-1} |v_{12}|^{1/2}  , \\
\mathcal{T}^{(3)}_n := \lambda_n \tilde B_{13}^{-1} (k_n) (v_{13})^{1/2}_- \bigl[ H_0 + k_n^2\bigr]^{-1} |v_{23}|^{1/2}  . 
\end{eqnarray}
The operators $Q_n$ are uniformly norm--bounded. The operators $ \mathcal{T}^{(1)}_n$ are also uniformly norm--bounded. Note that $\mathcal{T}^{(2)}_n = \mathcal{M'}^{(1)}_n  + \mathcal{M'}^{(2)}_n $, where 
$\mathcal{M'}^{(1,2)}_n$ is defined exactly as $\mathcal{M}^{(1,2)}_n$ except that $|v_{13}|$ gets replaced with $(v_{13})_-$. Thus from the above analysis 
it follows that $ \| \mathcal{T}^{(2)}_n \|$ is uniformly bounded. By similar arguments $ \| \mathcal{T}^{(3)}_n \|$ is uniformly bounded. 
Thus due to $\bigl\| |v_{ik}|^{1/2} |\psi_n| \bigr\| \to 0$ (see the no--clustering theorem in \cite{2}) the expression on the rhs of (\ref{psin32}) goes to zero in norm.
 \end{proof}

\begin{proof}[Proof of Theorem~\ref{th:2}]
Instead of (\ref{3bh3}) it suffices to prove that 
\begin{equation}\label{3bh}
 \left\| \hat \psi_n - \frac{\sqrt 2 e^{i\varphi_n}}{4\pi |\ln k_n|^{1/2}}
\frac{\chi_{[k_n, 1]} (|p_y|) e^{-|p_y| |x|}}{|x||p_y|} \right\| \to 0,
\end{equation}
where the hat denotes the action of the partial Fourier transform
$\mathcal{F}_{12}$, see Eq.~(17) in \cite{1}. Indeed, 
 computing explicitly the inverse Fourier transform
\begin{eqnarray}
 \fl  \frac{\sqrt 2}{(4\pi)  |\ln k_n|^{1/2}} \mathcal{F}_{12}^{-1} \left( \frac{\chi_{[k_n, 1]} (|p_y|) e^{-|p_y| |x|}}{|x||p_y|}\right) 
= \frac 1{2 \pi^{3/2} |\ln k_n|^{1/2}} 
\frac 1{|x||y|(|x|^2+|y|^2)} \nonumber\\
\fl \times \Bigl\{ e^{-|x|} \bigl[-|x|\sin |y| - |y|\cos|y|\bigr] 
- e^{-k_n|x|} \bigl[-|x|\sin(k_n |y|) - |y|\cos(k_n |y|)\bigr] \Bigr\} \label{17.04/1}
\end{eqnarray}
Now (\ref{3bh3}) follows directly
from (\ref{17.04/1}), (\ref{3bh})
after dropping those terms, whose norm goes to zero.

By Lemma~\ref{lem:1} $\|\psi_n - f^{(1)}_n\| \to 0$, where we have set $f^{(1)}_n :=
\bigl[H_0
+ k_n^2\bigr]^{-1} |v_{12}| \psi_n$.
From the Schr\"odinger equation for the term $\sqrt{|v_{12}|} \psi_n$ we obtain
\begin{eqnarray}
\sqrt{|v_{12}|} \psi_n =  - \Bigl\{1- \sqrt{|v_{12}|} \bigl(H_0 +
k_n^2\bigr)^{-1}\sqrt{|v_{12}|} \Bigr\}^{-1}
\sqrt{|v_{12}|} \nonumber \\
\times \bigl[H_0 + k_n^2\bigr]^{-1}\bigl(\lambda_n v_{13} + \lambda_n
v_{23}\bigr)\psi_n  . \label{this}
\end{eqnarray}
Substituting (\ref{this}) into the expression for $f^{(1)}_n$ results in
\begin{equation}\label{f_n}
f^{(1)}_n = \bigl[H_0 + k_n^2\bigr]^{-1} \sqrt{|v_{12}|} \Bigl\{1-
\sqrt{|v_{12}|} \bigl(H_0 + k_n^2\bigr)^{-1}\sqrt{|v_{12}|} \Bigr\}^{-1} \Phi_n
,
\end{equation}
where
\begin{equation}\label{Phi_n2}
\Phi_n := - \lambda_n \sqrt{|v_{12}|} \bigl[H_0 + k_n^2\bigr]^{-1}\bigl(v_{13} +
v_{23}\bigr)\psi_n  . 
\end{equation}
From the proofs of Lemmas~6, 9 in \cite{1} it follows that the operators $\sqrt{|v_{12}|} \bigl[H_0 + k_n^2\bigr]^{-1} \sqrt{|v_{s3}|} $ and 
$B^{-1}_{12} (k_n) \sqrt{|v_{12}|} \bigl[H_0 + k_n^2\bigr]^{-1} \sqrt{|v_{s3}|} $, where $B_{12}(k_n)$ is defined in Eqs.~(18)--(19) in
\cite{1}, are uniformly norm--bounded for $s=1,2$. 
Thus by (\ref{Phi_n2}) and Theorem~3 in \cite{2} $\| \Phi_n \| \to 0$ and $\|
B^{-1}_{12} (k_n) \Phi_n \| \to 0$. 
Acting with $\mathcal{F}_{12}$ on (\ref{f_n}) gives
\begin{equation}
\fl  \hat f^{(1)}_n  = \bigl[-\Delta_x + p_y^2 + k_n^2\bigr]^{-1} \sqrt{|v_{12}|}
\Bigl\{1- \sqrt{|v_{12}|} \bigl(-\Delta_x + p_y^2 +
k_n^2\bigr)^{-1}\sqrt{|v_{12}|} \Bigr\}^{-1} \hat \Phi_n  . 
\end{equation}
Because $\| \hat \Phi_n \| \to 0$ we can write
\begin{equation}
\hat f^{(1)}_n =  \hat f^{(2)}_n  + \hbox{o} (1) , 
\end{equation}
where
\begin{equation}
 \hat f^{(2)}_n := \chi_{[0, \rho_0]} \Bigl( \sqrt{p_y^2 +k_n^2}\Bigr) \hat
f^{(1)}_n ,
\end{equation}
and $\rho_0$ is a constant defined in Lemma~11 in \cite{1}.
Now using Lemma~11 in \cite{1} (see also discussion around Eq.~(111) in \cite{1}) we obtain
\begin{equation}
 \fl  \hat f^{(2)}_n =   \hat f^{(3)}_n + \chi_{[0, \rho_0]} \Bigl( \sqrt{p_y^2 +k_n^2}\Bigr) \mathcal{F}_{12} \mathcal{A}_{12} (k_n) \mathcal{F}^{-1}_{12} \mathcal{Z} \left( \sqrt{p_y^2 +k_n^2} \right) B^{-1}_{12} (k_n) \hat \Phi_n , 
\end{equation}
where $\mathcal{A}_{12} (k_n) := \bigl[H_0 + k_n^2\bigr]^{-1} \sqrt{|v_{12}|} B_{12} (k_n)$ and $\mathcal{Z}$ defined in \cite{1} remain uniformly norm--bounded for all $n$, 
see Lemmas~6, 11 in \cite{1}. The 
function $\hat f^{(3)}_n$ is defined as follows 
\begin{equation}
\hat f^{(3)}_n := \chi_{[0, \rho_0]} \Bigl( \sqrt{p_y^2 +k_n^2}\Bigr)
\bigl[-\Delta_x + |p_y|^2 + k_n^2\bigr]^{-1} \frac{\sqrt{|v_{12}|} }{a
\sqrt{|p_y|^2  + k_n^2}} \mathbb{P}_0 \hat \Phi_n , \label{xw11}
\end{equation}
where $a$ and $\mathbb{P}_0$ are defined in Eq.~(80) and Lemma~11 in \cite{1}. 
Therefore, since $\|
B^{-1}_{12} (k_n) \Phi_n \| \to 0$ 
\begin{equation}
  \hat f^{(2)}_n =   \hat f^{(3)}_n + \hbox{o}(1) . 
\end{equation}
It makes sense to introduce
\begin{equation}\label{gn}
 g_n (y) := \int d^3 x \phi_0 (x) \Phi_n (x, y) ,
\end{equation}
where $\phi_0$ was defined in Eq.~(77) in \cite{1}. The following inequality trivially follows from the exponential
bound on $V_{12}$ and the definition of $\phi_0$
\begin{equation}\label{tildephi1}
\phi_0 (x) \leq b'_1 e^{-b'_2 |x|} ,
\end{equation}
where $b'_{1,2} > 0$ are constants.
From the pointwise exponential fall off of $\psi_n$  it follows that
$g_n \in L^2 (\mathbb{R}^3)\cap L^1 (\mathbb{R}^3)$ for each $n$. We rewrite
(\ref{xw11}) with the help of (\ref{gn})
\begin{eqnarray}
 \fl  \hat f^{(3)}_n = \chi_{[0, \rho_0]} \Bigl( \sqrt{p_y^2 +k_n^2}\Bigr)
\bigl[-\Delta_x + p_y^2 + k_n^2\bigr]^{-1} \frac{\sqrt{|v_{12}|} \phi_0 (x) \hat
g_n (p_y) }{a \sqrt{p_y^2 + k_n^2}}    \nonumber\\
\fl  = \chi_{[0, \rho_0]} \Bigl( \sqrt{p_y^2 +k_n^2}\Bigr) \frac{\hat g_n (p_y)}{4\pi
a \sqrt{p_y^2 + k_n^2}} \int d^3 x' \frac{e^{-\sqrt{p_y^2 +
k_n^2}|x-x'|}}{|x-x'|} \phi_0  (x') |V_{12} (\alpha x') |^{1/2}, 
\end{eqnarray}
where $\alpha := \hbar /\sqrt{2\mu_{12}} $. Next, let us define
\begin{equation}
 \fl  \hat f^{(4)}_n := \chi_{[0, \rho_0]} \Bigl( \sqrt{p_y^2 +k_n^2}\Bigr)
\frac{\hat g_n (0)}{4\pi a \sqrt{p_y^2 + k_n^2}} \int d^3 x'
\frac{e^{-\sqrt{p_y^2 + k_n^2}|x-x'|}}{|x-x'|} \phi_0  (x') |V_{12}
(\alpha x') |^{1/2} , 
\end{equation}
where $ \hat g_n (0) \in \mathbb{C}$ is well-defined since $g_n \in  L^1
(\mathbb{R}^3)$ for each $n$. Using Lemma~\ref{lem:3} and the notation in (\ref{W(y)}) gives
\begin{eqnarray}
 \fl \| \hat f^{(4)}_n  - \hat f^{(3)}_n \|^2 \leq \int d^3 p_y \; \chi_{[0,
\rho_0]} \Bigl( \sqrt{p_y^2 +k_n^2}\Bigr)
\frac{c_n^2 |p_y|^{2\delta}}{16\pi^2 a^2 (p_y^2 + k_n^2)^{3/2}} \nonumber \\
\fl \times \int d^3 x' \int d^3 x'' W\Bigl(\sqrt{p_y^2 +k_n^2} (x'' - x' )\Bigr)
\phi_0  (x') |V_{12} (\alpha x')|^{1/2} \phi_0  (x'') |V_{12} (\alpha
x'')|^{1/2} \nonumber\\
\fl  \leq \frac{\vartheta^2 c_n^2 }{16\pi^2 a^2 }\int d^3 p_y \; \chi_{[0, \rho_0]} \Bigl( \sqrt{p_y^2 +k_n^2}\Bigr)
\frac{|p_y|^{2\delta}}{(p_y^2 + k_n^2)^{3/2}}  , 
\end{eqnarray}
where we used $W(s) \leq 2\pi$ and set 
\begin{equation}\label{vartheta}
 \vartheta  := \int d^3 x' \phi_0  (x') |V_{12} (\alpha x') |^{1/2} . 
\end{equation}
The constant in (\ref{vartheta}) is bounded, hence, by Lemma~\ref{lem:4} $ \| \hat f^{(4)}_n  - \hat f^{(3)}_n \|
\to 0$. As the next step we introduce
\begin{equation}
 \hat f^{(5)}_n := \chi_{[0, \rho_0]} \Bigl( \sqrt{p_y^2 +k_n^2}\Bigr)
\frac{R\Bigl( \sqrt{p_y^2 +k_n^2}\Bigr) \hat g_n (0)}{4\pi a \sqrt{p_y^2 +
k_n^2}} \frac{e^{-\sqrt{p_y^2 + k_n^2}|x|}}{|x|} ,
\end{equation}
where
\begin{equation}
R (s) :=  \int d^3 x' \frac{e^{-s|x'|}}{|x'|} \phi_0  (x') |V_{12}
(\alpha x') |^{1/2} . 
\end{equation}
Like in the proof of Theorem~\ref{th:1} we evaluate the square of the norm of the
difference
\begin{eqnarray}
 \| \hat f^{(5)}_n  - \hat f^{(4)}_n \|^2 \leq \int d^3 p_y \; \chi_{[0,
\rho_0]} \Bigl( \sqrt{p_y^2 +k_n^2}\Bigr)
\frac{|\hat g_n (0)|^2}{16\pi^2 a^2 (p_y^2 + k_n^2)^{3/2}} \nonumber \\
\times \int d^3 x' \int d^3 x'' \Bigl\{ W\Bigl(\sqrt{p_y^2 +k_n^2} (x'' - x'
)\Bigr)  + W(0) - W\Bigl(\sqrt{p_y^2 +k_n^2}x' \Bigr) \nonumber\\ 
  - W\Bigl(\sqrt{p_y^2
+k_n^2}x'' \Bigr)   \nonumber \Bigr\}
 \phi_0  (x') |V_{12}(\alpha x')|^{1/2}  \phi_0  (x'') |V_{12}
(\alpha x'') |^{1/2} \nonumber\\
\leq \frac{|\hat g_n (0)|^2 }{2\pi a^2 }\int d^3 p_y \;
\frac{\chi_{[0, \rho_0]} \Bigl( \sqrt{p_y^2 +k_n^2}\Bigr) }{(p_y^2 + k_n^2)}
\nonumber\\
\times \int d^3 x' \int d^3 x'' \; |x'| \phi_0  (x') |V_{12} (\alpha x')|^{1/2}
\phi_0  (x'') |V_{12} (\alpha x'')|^{1/2}  . 
\end{eqnarray}
On account of R1 and (\ref{tildephi1})
we conclude that $ \| \hat f^{(5)}_n  - \hat
f^{(4)}_n \| \to 0$ since $|\hat g_n (0)| \to 0$ by Lemma~\ref{lem:4}.
Observe that
\begin{equation}
 \left|R(s) - R(0)\right|\leq s \vartheta, 
\end{equation}
where $\vartheta$ is defined in (\ref{vartheta}). Hence, $ \| \hat f^{(6)}_n  - \hat f^{(5)}_n \| \to 0$, where by
definition
\begin{equation}
 \hat f^{(6)}_n := \chi_{[0, \rho_0]} \Bigl( \sqrt{p_y^2 +k_n^2}\Bigr)
\frac{R(0)\hat g_n (0)}{4\pi a \sqrt{p_y^2 + k_n^2}} \frac{e^{-\sqrt{p_y^2 +
k_n^2}|x|}}{|x|} . 
\end{equation}
Simplifying the argument of the exponential function we define
\begin{equation}
 \hat f^{(7)}_n := \chi_{[0, \rho_0]} \Bigl( \sqrt{p_y^2 +k_n^2}\Bigr)
\frac{R(0)\hat g_n (0)}{4\pi a \sqrt{p_y^2 + k_n^2}} \frac{e^{-|p_y||x|}}{|x|} . 
\end{equation}
After straightforward calculation we obtain
\begin{eqnarray}
 \| \hat f^{(7)}_n  - \hat f^{(6)}_n \|^2 = \int d^3 p_y \; \chi_{[0, \rho_0]}
\Bigl( \sqrt{p_y^2 +k_n^2}\Bigr)
\frac{R^2 (0)|\hat g_n (0)|^2}{4\pi a^2 (p_y^2 + k_n^2)} \nonumber \\
\times \left[\frac 1{2\sqrt{p_y^2 + k_n^2}} + \frac 1{2|p_y|} - \frac
2{\sqrt{p_y^2 + k_n^2} + |p_y|}\right]  . 
\end{eqnarray}
Replacing in the last fraction $|p_y|$ with $\sqrt{p_y^2 + k_n^2}$ results in
the following inequality
\begin{equation}
 \fl  \| \hat f^{(7)}_n  - \hat f^{(6)}_n \|^2 \leq \frac{R^2 (0)|\hat g_n
(0)|^2}{8\pi a^2 }\int d^3 p_y \;
\frac{\chi_{[0, \rho_0]} \Bigl( \sqrt{p_y^2 +k_n^2}\Bigr) }{(p_y^2 + k_n^2)}
\left[ \frac 1{|p_y|} - \frac 1{\sqrt{p_y^2 + k_n^2}}\right]  . 
\end{equation}
The integrals can be calculated explicitly, see \cite{tables}, which results in $\| \hat f^{(7)}_n  - \hat f^{(6)}_n \| \to 0$.
At last, we simplify the expression setting
\begin{equation}
 \hat f^{(8)}_n := \frac{R(0)\hat g_n (0)}{4\pi a }  \frac{\chi_{[k_n, 1]}
\bigl(|p_y|\bigr)e^{-|p_y||x|}}{|x||p_y|} . 
\end{equation}
Again, one easily finds that $\| \hat f^{(8)}_n  - \hat f^{(7)}_n \| \to 0$.
Summarizing, we have $\| \hat f^{(i+1)}_n  - \hat f^{(i)}_n \| \to 0$ for $i= 1,
\ldots, 7$.
Thus from $\| \hat \psi_n -  \hat f^{(1)}_n \| \to 0 $ it follows that $\| \hat
\psi_n -  \hat f^{(8)}_n \| \to 0 $. Using that  $\| \hat \psi_n \| = 1$ we obtain
(\ref{3bh}).
 \end{proof}

\begin{lemma}\label{lem:3}
There exists a sequence  $c_n >0 $,  $c_n \to 0$ such that
\begin{equation}
 \bigl|\hat g_n (p_y) - \hat g_n (0)\bigr| \leq c_n |p_y|^{\delta} ,
\end{equation}
where $\delta$ is defined in (\ref{restr}).
\end{lemma}
\begin{proof}
The trivial inequality $|e^{ip_y\cdot y} - 1|\leq
|p_y|^{\delta}|y|^{\delta}$ implies that
\begin{eqnarray}
  \bigl|\hat g_n (p_y) - \hat g_n (0)\bigr| \leq \int d^3 y  \bigl|e^{ip_y\cdot
y} - 1\bigr| |g_n(y)|\leq |p_y|^{\delta} c_n ,
\end{eqnarray}
where $c_n =  \int d^3 y |y|^{\delta} |g_n (y)|$ goes to zero by Lemma~\ref{lem:4}.
 \end{proof}
The following lemma makes use of the absence of zero energy resonances in particle
pairs $\{1,3\}$ and $\{2,3\}$.
\begin{lemma}\label{lem:4}
 The sequence $c_n =  \int d^3 y \bigl( 1 + |y|^{\delta}\bigr) |g_n (y)|$ is well-defined and
goes to zero.
\end{lemma}
\begin{proof}
By definitions (\ref{gn}) and (\ref{Phi_n2}) we have $|g_n(y)| \leq |g^{(1)}_n(y)|
 + |g^{(2)}_n(y)|$, where
\begin{eqnarray}
 g^{(1)}_n(y) := \lambda_n \int d^3 x \, \phi_0 |v_{12}|^{1/2} \bigl[H_0 + k_n^2 \bigr]^{-1}
v_{13}\psi_n , \\
 g^{(2)}_n(y) := \lambda_n \int d^3 x \, \phi_0 |v_{12}|^{1/2} \bigl[H_0 + k_n^2 \bigr]^{-1}
v_{23}\psi_n  . 
\end{eqnarray}
Consequently $c_n \leq c_n^{(1)} + c_n^{(2)}$, where 
\begin{equation}\label{xw15}
 c_n^{(i)} := \int d^3 y \bigl(1 + |y|^{\delta}\bigr) |g^{(i)}_n(y) |. 
\end{equation}
Below we shall prove that $c_n^{(1)} \to 0$, the fact that 
$c_n^{(2)} \to 0$ is proved analogously. Let us mention that appearing below integrals and interchanged oder of integration 
can be easily justified using the pointwise exponential fall off of $\psi_n$ \cite{reed}. 

 We have
\begin{equation}
|g^{(1)}_n(y)| \leq \int d^3 x \,  \bigl| V_{12} (\alpha x)\bigr|^{1/2} \phi_0 (x) |F_n|(x,y),
\end{equation}
where $F_n$ was defined in (\ref{xwsplit}). 
On account of R1 and (\ref{tildephi1}) it follows that 
\begin{equation}\label{xw13}
|g^{(1)}_n(y)| \leq \tilde b_1 \int d^3 x  e^{-\tilde b_2 |x|}  |F_n|(x,y) , 
\end{equation}
where $\tilde b_{1,2} > 0$ are constants.
Using (\ref{xwsplit}) and (\ref{18.04/1})--(\ref{18.04/2}) gives 
\begin{eqnarray}
 |F_n| \leq \sum_{i=1}^3 |F^{(i)}_n| \leq \sum_{i=1}^3 \tilde F^{(i)}_n  , \label{xw14}\\
\tilde F^{(i)}_n := \bigl[ H_0 + k_n^2\bigr]^{-1} |v_{13}|^{1/2}
\tilde B_{13}(k_n) \Psi^{(i)}_n  . \label{xw17}
\end{eqnarray}
Substituting (\ref{xw13}), (\ref{xw14}) into (\ref{xw15}) we obtain
\begin{equation}\label{xw22}
\fl c_n^{(1)}  
\leq  \tilde b_1 \sum_{i=1}^3 \int d^3 \eta \; d^3 \zeta  \; \Bigl( 1 + \bigl|m_{y\eta}
\eta +  m_{y\zeta} \zeta \bigr|^{\delta}\Bigr) e^{-\tilde b_2 |m_{x\eta} \eta +
m_{x\zeta} \zeta|}  \tilde F^{(i)}_n (\eta , \zeta) . 
\end{equation}
Let us consider the term $ \tilde F^{(i)}_n (\eta , \zeta) $. Acting on it with 
direct and inverse partial Fourier
transforms (\ref{xw20}) we get 
\begin{equation}\label{8.11}
\tilde F^{(i)}_n = \mathcal{F}_{13}^{-1} \bigl[ -\Delta_\eta + p_\zeta^2 + k_n^2 \bigr]^{-1} |v_{13}|^{1/2} t_n (p_\zeta)\hat \Psi^{(i)}_n , 
\end{equation}
where $\hat \Psi^{(i)}_n  = \mathcal{F}_{13}  \Psi^{(i)}_n $. This can be explicitly rewritten as 
\begin{equation}
 \fl  \tilde F^{(i)}_n (\eta , \zeta)
= \frac 1{2^{7/2} \pi^{5/2}} \int d^3 \eta' d^3 p_\zeta \; e^{ip_\zeta \cdot \zeta}\:
\bigl| V_{13}(\alpha' \eta')\bigr|^{1/2}   \frac{e^{-\sqrt{p_\zeta^2 +
k_n^2}|\eta-\eta'|}}{|\eta-\eta'|} \tilde t_n (p_\zeta)\hat \Psi^{(i)}_n (\eta',
p_\zeta) . 
\end{equation}
Hence,
\begin{equation}\label{xw23}
 \fl  \bigl| \tilde F^{(i)}_n (\eta , \zeta) \bigr| \leq \frac 1{2^{7/2} \pi^{5/2}}  \int d^3
\eta' d^3 p_\zeta \; \bigl| V_{13}(\alpha' \eta')\bigr|^{1/2}
\frac{e^{-\sqrt{p_\zeta^2 + k_n^2}|\eta-\eta'|}}{|\eta-\eta'|} \tilde t_n
(p_\zeta) \bigl| \hat \Psi^{(i)}_n (\eta', p_\zeta) \bigr| . 
\end{equation}
Substituting (\ref{xw23}) into (\ref{xw22}) and interchanging the order of integration we
obtain the inequality
\begin{equation}\label{dday5}
\fl  c_n^{(1)} \leq \frac {\tilde b_1}{2^{7/2} \pi^{5/2}}
\sum_{i=1}^3
\int d^3 \eta' \int d^3 p_\zeta \; \bigl| V_{13}(\alpha' \eta')\bigr|^{1/2}
\tilde t_n (p_\zeta) \bigl| \hat \Psi^{(i)}_n (\eta', p_\zeta) \bigr| J (\eta',
p_\zeta)  ,
\end{equation}
where we define
\begin{equation}\label{xw34}
 \fl  J (\eta', p_\zeta) :=  \int d^3 \eta \int d^3 \zeta  \;
 \frac{e^{-\sqrt{p_\zeta^2 + k_n^2}|\eta-\eta'|}}{|\eta-\eta'|} \Bigl( 1+  \bigl|m_{y\eta}
\eta +  m_{y\zeta} \zeta \bigr|^{\delta} \Bigr) e^{-\tilde b_2 |m_{x\eta} \eta +
m_{x\zeta} \zeta|} . 
\end{equation}
Applying the Cauchy--Schwarz inequality to (\ref{dday5}) gives 
\begin{equation}\label{dday6}
\fl  c_n^{(1)} \leq \frac {\tilde b_1 }{2^{7/2} \pi^{5/2}}
\sum_{i=1}^3 \bigl\| \Psi^{(i)}_n \bigr\|
\left( \int d^3 \eta' \int d^3 p_\zeta \; \bigl| V_{13}(\alpha' \eta')\bigr| \,
{\tilde t}^{\; 2}_n (p_\zeta) \, J^2 (\eta', p_\zeta) \right)^{1/2}  .
\end{equation}
Inserting the estimate from Lemma~\ref{lem:5} we finally get
\begin{eqnarray}
\fl  c_n^{(1)} \leq \frac {\tilde b_1
c\sqrt{C}}{ 2^{5/2} \pi^2 } \sum_{i=1}^3 \bigl\| \Psi^{(i)}_n \bigr\|
\left( \int_0^1 \frac{s^2 \bigl(s^{1-\delta} + k_n^{1-\delta}\bigr)^2}{\bigl(s^2
+ k_n^2\bigr)^{2+\delta}}ds  + \int_1^\infty \frac{s^2 \bigl(1 + k_n^{1-\delta}\bigr)^2}{\bigl(s^2 +
k_n^2\bigr)^2} ds\right)^{1/2} \label{xw25} .
\end{eqnarray}
where $C := \int d^3 \eta' \bigl| V_{13}(\alpha'
\eta')\bigr|(1+|\eta'|)^{2\delta} $ is finite by (\ref{restr}). The last
integral in (\ref{xw25}) is clearly uniformly bounded for all $n$.
To see that the first integral in (\ref{xw25}) is uniformly bounded we use the following
inequality
\begin{equation}
 (s^{1-\delta} +k_n^{1-\delta})^2 \leq 2 (s^{1-\delta})^2 + 2 (k_n^{1-\delta})^2
\leq 4 (s^2 + k_n^2)^{1-\delta} , 
\end{equation}
where we used $a^\alpha + b^\alpha \leq 2(a + b)^\alpha$ for any $a,b \geq 0$ and $0 \leq \alpha \leq 1$. 
Hence,
\begin{equation}
 \int_0^1 \frac{s^2 \bigl(s^{1-\delta} + k_n^{1-\delta}\bigr)^2}{\bigl(s^2 +
k_n^2\bigr)^{2+\delta}}ds \leq 4\int_0^1
\frac{s^2 ds}{\bigl(s^2 + k_n^2\bigr)^{1+2 \delta}} \leq 4\int_0^1
\frac{s^2 }{s^{2+4\delta}}ds \leq 8 . 
\end{equation}
Thus the rhs of (\ref{xw25}) goes to zero by Lemma~\ref{lem:2}. 
 \end{proof}

\begin{lemma}\label{lem:5}
 The following estimates hold
\begin{eqnarray}
J (\eta', p_\zeta) \leq \frac{c (1+|\eta'|)^{\delta}}{p_\zeta^2 + k_n^2} \quad
\quad \mathrm{for} \;\;\; |p_\zeta| \geq 1 , \\
J (\eta', p_\zeta) \leq \frac{c (1+|\eta'|)^{\delta}}{\bigl(p_\zeta^2 +
k_n^2\bigr)^{1+\delta/2}} \quad \quad \mathrm{for} \;\;\; |p_\zeta| \leq 1  , 
\end{eqnarray}
where $c > 0$ is a constant.
\end{lemma}
\begin{proof}
Using the trivial inequality $|z+z'|^{\delta} \leq |z|^{\delta} + |z'|^{\delta}$
for any $z,z' \in \mathbb{R}^3$ it is easy to see that
\begin{equation}\label{dday1}
 \int d^3 \zeta \; \Bigl( 1 + \bigl|m_{y\eta} \eta +  m_{y\zeta} \zeta \bigr|^{\delta} \Bigr) 
e^{-\tilde b_2 |m_{x\eta} \eta +  m_{x\zeta} \zeta|} \leq c' (1+|\eta|)^{\delta} , 
\end{equation}
where $c' > 0$ is some constant. Using (\ref{xw34}) and (\ref{dday1}) we obtain
\begin{eqnarray}
 \fl J (\eta', p_\zeta) \leq c' \int d^3 \eta \;  \frac{e^{-\sqrt{p_\zeta^2 +
k_n^2}|\eta-\eta'|}}{|\eta-\eta'|} (1+|\eta|)^{\delta} 
\leq c' \int d^3 t \;  \frac{e^{-\sqrt{p_\zeta^2 + k_n^2}|t|}}{|t|}
 (1+|t+\eta'|)^{\delta} \nonumber\\
\fl \leq
 c' \int d^3 t \;  \frac{e^{-\sqrt{p_\zeta^2 + k_n^2}|t|}}{|t|}   \bigl\{ 1
+|\eta'|^{\delta}+|t|^{\delta}\bigr\}  . 
\end{eqnarray}
Now the statement easily follows. 	
 \end{proof}

\begin{remark}
The proof of Lemma~\ref{lem:3} is not just a mathematical formality, as can be illustrated by the following example. Suppose that two particles $2,3$ are identical and 
$V_{ik} \leq 0$ for all $1 \leq i < k \leq 3$. Suppose also that $H \equiv H(1)\geq 0$, where $H(1)$ is defined through (\ref{hami}), and particle pairs $\{1,2\}$ and $\{1,3\}$ 
have zero energy resonances. In this case there exists \cite{yafaev,sobol,8.efimov} an orthonormal sequence $\phi_n$ such that $H \phi_n = E_n \phi_n$, where $E_n < 0, E_n \to 0$. Similar to Lemma~\ref{lem:1} 
one can prove that $\phi _n = f_n^{(12)} + f_n^{(13)}$, where $f_n^{(ik)} := [H_0 + k_n^2]^{-1} |v_{ik}|\phi_n$ and the sequences $f_n^{(12)},  f_n^{(13)}$ must totally spread. 
However, a relation like (\ref{3bh3}) for $f_n^{(12)}$ (and a similar relation for $f_n^{(13)}$ with rotated Jacobi coordinates) would be wrong. Indeed, as we have already mentioned in the remark after Theorem~\ref{th:2} 
one can choose the mass ratios in such a way that the sequence $\phi_n$ would have an angular momentum different from zero. At the same time,  the limiting expression in (\ref{3bh3}) always has zero angular momentum. 
Additionally, one can prove that $(\phi_n, \phi_{n+1}) = 0$ would not hold in this case in the limit of large $n$. This example demonstrates that the condition that only 
one particle pair has a zero energy resonance is crucial to the proof of Lemma~\ref{lem:3}.  
\end{remark}

Finally, let us show how the angular probability distribution in (\ref{sony}) can be derived using a less rigorous but more physical approach. The derivation below was proposed 
by one of the referees, whose contribution is gratefully acknowledged. Suppose that the interaction between particles $1,2$ depends on $|x|$ and is resonant, while other 
pair--interactions are non--resonant. 
Let us consider the ground state wave function $\psi_\infty (x,y) > 0$ of the Hamiltonian (\ref{hami}) for $\lambda=\lambda_{cr}$, which as we know from \cite{1} is not 
normalizable. The wave function $\psi_\infty$ obeys the equation $\bigl[H_0 + v_{12} + v_{13} + v_{23}\bigr] \psi_\infty = 0$, where the interactions $v_{13}, v_{23}$ 
can be dropped because they are non--resonant (c.f.  Lemma~\ref{lem:1}). Since the rest term in the Hamiltonian is invariant with respect to independent rotations of 
vectors $x$ and $ y$, the ground state should possess the same symmetry, that is, we can write the wave function as $\psi_\infty (|x|,|y|)$. Following the recipe in \cite{8.efimov,peierls} 
we can replace the resonant interaction $v_{12}$ through the boundary condition $\partial(|x|\psi_\infty)/\partial|x| = 0$ and solve instead the equation $H_0 \psi_\infty = 0$ 
using this boundary condition. Setting $\psi_0 (|x|,|y|) := |x||y|\psi_\infty (|x|,|y|)$ we obtain the following equation 
\begin{equation}
 \left( \frac{\partial^2}{\partial |x|^2} + \frac{\partial^2}{\partial|y|^2} \right) \psi_0 (|x|,|y|) = 0 , \label{17:17:1}
\end{equation}
where $\psi_0 (|x|,|y|)$ should satisfy boundary conditions $\partial \psi_0/\partial |x| = 0$ and $\psi_0 (|x|, 0) = 0$. In polar coordinates (\ref{17:17:1}) reads 
\begin{equation}
 \frac 1{\rho} \frac{\partial  \psi_0 (\rho, \theta)}{\partial \rho} + \frac{\partial^2  \psi_0 (\rho, \theta)}{\partial \rho^2} +  
\frac 1{\rho^2} \frac{\partial^2  \psi_0 (\rho, \theta)}{\partial \theta^2} = 0 , \label{ref10.1;1}
\end{equation}
where $\rho, \theta$ were defined in Sec.~\ref{sec:1}. Separating radial and angular variables 
one easily finds that the solution of (\ref{ref10.1;1}), which satisfies the aforementioned boundary conditions, is given by $\rho^{-n} \sin (n \theta)$ for $n= 1, 3, 5, \ldots$. The non--normalizable wavefunction corresponds to $n=1$, which gives 
$ \psi_0 (\rho, \theta) = \rho^{-1} \sin (\theta)$. Returning back to the original 
wave function $\psi_\infty$  results in $\psi_\infty (\rho, \theta) =[\rho^{-3} \cos\theta]^{-1}$. This is the expression in (\ref{17:23:1}) 
that we obtain after removing the normalization factor and setting $k_n = 0$. This angular dependence in $\psi_\infty (\rho, \theta)$ leads to the universal angular 
probability distribution (\ref{sony}). 

\section{Physical Applications}\label{sec:phys} 

 In nuclear physics one encounters nuclei \cite{vaagen1}, which effectively
possess the three--particle Borromean structure consisting of two neutrons and a tightly bound 
core. In most applications the core can be well treated as a structureless particle.  
Borromean in this context means that the three constituents are pairwise unbound rather like
heraldic Borromean rings. The ground states in some of these nuclei are weakly bound and two neutrons form
a dilute halo around the core. Thereby a substantial part of the wavefunction is located in the classically forbidden region so that resulting 
inter-particle distances exceed by far the range of the interaction. 
Typical examples of such halo nuclei are weakly bound $^6$He and $^{11}$Li. 
The calculated density correlation plots in \cite{vaagen1,vaagen2} 
reveal the formation of the so--called ``dineutron peak'' in the ground state. There is another peak called a cigar--like peak but 
the substantial part of the wave function that is responsible for the halo formation 
concentrates in the dineutron peak. The dineutron peak is remarkably well fitted by the angular probability distribution in (\ref{sony}). 

Additional applications one could find in Efimov physics. 
The so--called three--particle Efimov states predicted in \cite{8.efimov} appear when  two 
binary subsystems either have very large scattering lengths or bound states close to zero energy threshold. 
Efimov states were found experimentally in the ultracold Bose gas of cesium atoms \cite{8.kraemer}. In \cite{1} we predicted the existence 
of very spatially extended halo-like states for three atoms near zero-energy threshold, if one pair of atoms has a
large scattering length (that is, it is close to the zero-energy resonance). These states can be looked for in ultracold gas mixtures prepared
through the appropriate Feshbach tuning. The reported result shows that the density distribution in such system of three atoms would have a universal form 
described by (\ref{sony}), which at sufficiently large distance should match the nucleon density in nuclear halos. 

\ack 

The author would like to thank Prof. Walter Greiner for the warm hospitality at FIAS.

\end{document}